\documentclass[11pt,letterpaper,onecolumn,oneside,notitlepage,fleqn]{article}
\usepackage{setspace}
\usepackage[letterpaper, margin=1in]{geometry}

\usepackage[dvipsnames]{xcolor}
\usepackage{color, colortbl}
\usepackage[utf8]{inputenc}
\usepackage[english]{babel}
\usepackage{listings}
\usepackage{amsmath}
\usepackage{amssymb}
\usepackage{amsthm}
\usepackage{faktor}
\usepackage{stmaryrd}
\usepackage{semantic}
\usepackage{multirow}

\usepackage{graphicx}
\usepackage{hyperref}

\usepackage{paralist}
\usepackage{nicefrac}
\usepackage{listings}

\usepackage{pdfpages}
\usepackage{comment}
\usepackage{wasysym}
\usepackage{leftidx}

\usepackage{rotating}

\usepackage{booktabs}
\usepackage{subcaption}

\usepackage{tikz}
\usetikzlibrary{shapes}
\usetikzlibrary{fit}
\usepackage{pgf}
\tikzset{node distance=2cm, text width=1cm, align=center, auto}

\usepackage[shortcuts]{extdash}

\usepackage{lstcoq}

\usepackage{titlesec}

\newcommand{\includeframe}[4]{\makebox[#2\linewidth]{\includegraphics[page=#1,width=#2\linewidth,trim=0cm 0cm 0cm 0cm,clip=true,#3]{#4}}}

\newtheorem{theorem}{Theorem}

\newtheorem{example}{Example}
\newtheorem{remark}{Remark}
\newtheorem{proposition}{Proposition}
\newtheorem{definition}{Definition}
\newtheorem{lemma}{Lemma}

\newcommand{\bdd}{BDD}
\newcommand{\fdd}{\texttt{FDD}}
\newcommand{\sdd}{SDD}

\newcommand{\robdd}{ROBDD}
\newcommand{\robddn}{ROBDD+N}
\newcommand{\zdd}{ZDD}
\newcommand{\tbdd}{TBDD}
\newcommand{\esr}{\texttt{ESRBDD}}
\newcommand{\esrlz}{\esr$-\texttt{L}_0$}
\newcommand{\chainD}{\texttt{Chain-DD}}
\newcommand{\chainB}{\texttt{Chain-BDD}}
\newcommand{\chainZ}{\texttt{Chain-ZDD}}

\newcommand{\ldd}{\lambda\texttt{DD}}
\newcommand{\ls}{\ldd\texttt{-S}}
\newcommand{\ldds}{\ldd\texttt{-S}}
\newcommand{\lsn}{\ldd\texttt{-S-N}}

\newcommand{\lddoU}{\ldd\texttt{-O-U}}
\newcommand{\lddo}{\ldd\texttt{-O}}
\newcommand{\lddoNU}{\ldd\texttt{-O-NU}}
\newcommand{\lddoCuz}{\ldd\texttt{-O-C}_{10}}
\newcommand{\lddoUCuz}{\ldd\texttt{-O-UC}_{10}}
\newcommand{\lddoNUCuzuu}{\ldd\texttt{-O-NUC}_{10}\texttt{C}_{11}}
\newcommand{\lddoUCz}{\ldd\texttt{-O-UC}_{0}}
\newcommand{\lddoNUC}{\ldd\texttt{-O-NUC}}
\newcommand{\lddoUC}{\ldd\texttt{-O-UC}}
\newcommand{\lddoNUCX}{\ldd\texttt{-O-NUCX}}

\newcommand{\N}{\mathbb{N}}
\newcommand{\B}{\mathbb{B}}

\newcommand{\F}[1]{{\B^{{#1}\to 1}}}

\newcommand{\zero}{\mathit{0}}
\newcommand{\un}{\mathit{1}}

\renewcommand{\tt}{\mathsf{t}}

\newcommand{\uu}{\mathsf{u}}
\newcommand{\xx}{\mathsf{x}}

\newcommand{\cc}{\mathsf{c}}

\newcommand{\xor}{\oplus}

\newcommand{\abs}[1]{\lvert#1\rvert}
\newcommand{\semb}[1]{\llbracket#1\rrbracket}

\newcommand{\DAGaml}{\texttt{DAGaml}}

\newcommand{\camlarrow}{\longrightarrow\hspace{-0.4em}\blacktriangleright}

\newcommand{\nc}[1]{{{#1}^{\natural}}}
\newcommand{\ncn}[1]{{{#1}^{\natural}_{\bullet}}}

\begin{document}

\title{Ordered Functional Decision Diagrams: \\
A Functional Semantics For Binary Decision Diagrams}

\author{Joan Thibault \\ \href{mailto:joan.thibault@inria.fr}{joan.thibault@inria.fr} 
   \and Khalil Ghorbal \\ \href{mailto:khalil.ghorbal@inria.fr}{khalil.ghorbal@inria.fr} }

\maketitle

\begin{abstract}
We introduce a novel framework, termed $\ldd$, that revisits Binary Decision Diagrams from a purely functional point of view.
The framework allows to classify the already existing variants, including the most recent ones like \chainD{} and \esr{}, as implementations of a special class of ordered models. 
We enumerate, in a principled way, all the models of this class and isolate its most expressive model.
This new model, termed $\lddoNUCX$, is suitable for both dense and sparse Boolean functions, and is moreover \emph{invariant} by negation.
The canonicity of $\lddoNUCX$ is formally verified using the Coq proof assistant.
We furthermore give bounds on the size of the different diagrams: the potential gain achieved by more expressive models can be at most linear in the number of variables $n$. 
\end{abstract}

\section*{Introduction}
A Binary Decision Diagram (\bdd{}) is a versatile graph-based data structure, well suited to effectively represent and manipulate Boolean functions.
As shown by Bryant~\cite{Bryant1986}, although a Binary Decision Diagram (\bdd{}) has an exponential worst-case size, many practical applications yield more concise representations thanks to the elimination of \emph{useless} nodes, i.e., those nodes having their outgoing edges pointing towards the same subgraph.
Many \bdd{} variants~\cite{BurchLong1992,MinatoVariants,RolfVariantDual} have been subsequently designed to capture specific application-dependent properties in order to further reduce the size of the diagram or to efficiently perform specific operations.
For instance, Zero-suppressed Decision Diagrams~\cite{IntroZDD,ZDD}, or \zdd{}, form a notable variant that is well suited to encode sparse functions, i.e., functions that evaluate to zero except for a limited number of valuations of their inputs.

Most recently, two new variants, namely \chainB{}~\cite{CBDD} and \esr{}~\cite{ESRBDD-TACAS2019}, propose to combine \robdd{} and \zdd{} in order to get a data structure suitable for both dense and sparse functions.
In this work, we are also interested in combining existing variants in order to benefit from their respective sweet spots.

Our approach is however, drastically different: we combine reduction rules by composing their \emph{functional abstraction} (or interpretation).
To do so, we introduce a new functional framework, together with its related data structure, that we term $\ldd$.
Special variables, like useless variables, are captured by elementary operators (or functors) acting on Boolean functions. 
We exemplify our approach by considering the so-called \emph{canalizing} variables~\cite{canalising}, which form an important class of special variables dual to useless variables: their valuation fixes the output of the function regardless of the valuation of the other variables.
\emph{In our framework, designing a data structure that captures several special variables amounts to combining, at the functional level, various elementary operators, while paying attention to their possible interactions.}

The functional framework allows not only to compare the expressive power of the modeled variants, but also, and more importantly, to design in a principled way new models with higher compression rates.
We present in particular a new canonical data structure, termed $\lddoNUCX$, that combines canalizing and useless variables while
supporting negation, unlike \zdd{}, \chainB{}, and \esr{}.
The obtained graphs are \emph{invariant} by negation, i.e., the diagram of the negation of a function differs from the diagram of the function itself by only appending the symbol that encodes negation.
As a consequence, negation is a constant-time operation. 

 
The three main contributions of this paper can be summarized as follows. 
(I) A general functional framework for Boolean functions relying on the Shannon combinator (Section~\ref{sec:omodels}) for a class of \emph{ordered} models, denoted $\lddo$, classifying many already existing \bdd{} variants (Section~\ref{sec:dd-unify}).
(II) A new model, called $\lddoNUCX$ (Section \ref{sec:onucx}), supporting all the primitives defining the class $\lddo$, including negation (Section \ref{sec:negation}). 
$\lddoNUCX$ is a strict generalization of the recent variants $\chainD$ and $\esr$.
(III) A comparison of the number of nodes showing that the potential gain can be at most linear in $n$ regardless of the expressiveness of the models. 

\begin{remark}
All results and theorems in the upcoming sections are part of the formalization project of the $\ldd$ data structure in the Coq proof assistant. 
For lack of space, the formalization details wont be detailed. 
\end{remark}

\section{Preliminaries}\label{sec:bf}
A Boolean function of arity $n \in \N$ is a form (or functional) from $\B^n$ to $\B$.
It operates on an ordered tuple of Booleans of dimension $n$, $(x_{0},\dotsc,x_{n-1})$,
by assigning a Boolean value to each of the $2^n$ valuations of its tuple.
The set of Boolean functions of arity $n$, denoted by $\F{n}$, is thus finite and contains $2^{2^n}$ elements.
In particular, $\F{0}$ has two elements and is isomorphic to $\B$ itself (only the types differ: functions on the one hand, and co-domain elements, or Booleans, on the other hand).
To avoid confusion, we use a different font for functions: $\zero$ will denote the constant function of arity zero returning $0$, and $\un$ will denote the constant function of arity zero returning $1$.

We rely on a binary non-commutative operator (or functor), sometimes referred to as the Shannon operator in the literature, defined as follows.
\begin{definition}[Shannon operator]\label{def:star}
    Let $f,g$ be two Boolean functions defined over the same set of variable $x_1,\dotsc,x_n$. 
    The \emph{Shannon operator} $\star: \F{n} \times \F{n} \to \F{n+1}$ id defined as  
\[
    (f \star g): (x_0,x_{1},\dotsc,x_n) \mapsto (\lnot x_0 \land f(x_{1},\dotsc,x_n)) \lor (x_0 \land g(x_{1},\dotsc,x_n)) \enspace .
\]
\end{definition}
The Shannon operator is (i) \emph{universal}: any Boolean function can be fully decomposed, by induction over its arity, all the way down to constant functions; and (ii) \emph{elementary}:
it operates on two functions defined over the same set of variables and increases the arity by exactly one.
In our definition, this is done by appending a new variable, $x_0$, at position $0$, to the ordered tuple $(x_1,\dotsc,x_n)$.

Depending on the operands $f$ and $g$, this newly introduced variable can be of several \emph{types}.
Two kinds of variables will be of particular interest in this paper. 
\begin{definition}[Useless Variable]
    \label{def:uselessvar}
    For any function $f$, the newly added variable in $f \star f$ is said to be \emph{useless}. 
\end{definition}
\begin{definition}[Canalizing Variable]
    \label{def:canalizingvar}
    Let $f$ denote a function and let $c$ denote a constant function ($\zero$ or $\un$) having the same arity of $f$.  
    The newly added variable in $c \star f$ or $f \star c$ is said to be \emph{canalizing}.
\end{definition}
For instance, in $g: (x_0,x_1) \mapsto x_0 \land \neg x_1$, the variable $x_0$ is canalizing.
Indeed, let $f: x_1 \mapsto \lnot x_1$, then $g = \zero \star f$.
Canalizing variables are dual to useless variables in the sense that their valuations could fix the output of the entire function.
For the function $g$, one has $g(0,x_1) = 0$ regardless of $x_1$.

A key observation for useless and canalizing variables alike is that the Shannon operator acts on \emph{one} function $f$ and produces a new function by appending a typed variable (useless or canalizing) to the ordered list of inputs of $f$. 
As such, the binary operator behaves like a unary operator acting on Boolean functions. 
This simple observation is at the heart of the functional framework introduced next.

\section{Ordered Functional Decision Diagrams}\label{sec:omodels}
We introduce a new data structure, akin to ordered \bdd{}, that we term Ordered Functional Decision Diagram or $\ldd\texttt{-O}$.  

\subsection{Syntax and Semantics}\label{sec:oldd}
Let $\Delta$ denote a finite set of letters, and $\Delta^{*}$ denote the set of all words (freely) generated by concatenating any finite number of letters from $\Delta$. In particular, $\varepsilon \in \Delta^{*}$ denotes the empty word.
\begin{definition}[$\ldd\texttt{-O}$]
    A \emph{$\ldd\texttt{-O}$} $(\phi,n)$, $n \in \mathbb{N}$, is a parameterized recursive data structure defined as follows
    \begin{align*}
        (\phi,0) &:= (\zero,0) \mid (\un,0)     \\
        (\phi,n+1) &:= \ell (\phi,n) \mid (\phi,n) \diamond (\phi,n),
    \end{align*}
    where $\ell$ is a letter in $\Delta$.
    The letters in $\Delta$ as well as the binary operator $\diamond$ increase the parameter $n$ by exactly one.
    Notice that the two operands of $\diamond$ have necessarily the same parameter $n$.
\end{definition}
We drop the arity $n$ from the notation whenever clear from the context.  
A $\ldd\texttt{-O}$ can be represented as a directed acyclic graph, with (all) edges labeled with words in $\Delta^{*}$.
The representation requires three types of nodes: one root node ($\blacktriangledown$), two terminal nodes ($\square$ and $\blacksquare$), and a diamond node ($\Diamond$).
The structure is defined inductively as follows.
\begin{itemize}
    \item a root node pointing to a terminal node;
    \item or a root node pointing to a $\ldd\texttt{-O}$ graph;
    \item or a root node pointing to a diamond node having two outgoing edges, each of which pointing to a $\ldd\texttt{-O}$.
\end{itemize}
An edge can have only one label: composing $\xrightarrow{w_1}$ and $\xrightarrow{w_2}$ results in $\xrightarrow{w_1.w_2}$, where $w_1.w_2$ denotes the concatenation of $w_1$ and $w_2$.

The $\ldd\texttt{-O}$ data structure, or equivalently its graph representation, can be given, by induction, a semantics over Boolean functions.
An elementary unary operator acts on a Boolean function $f$ of arity $n \geq 0$ by appending a typed variable to the input of $f$ (regarded as an ordered tuple). For instance, the elementary operator $\delta_{\uu}: f \mapsto f \star f$ appends a useless variable to the input of $f$.
Similarly, each of the following elementary operators append a different canalizing variable
to the input of $f$. 
\begin{align*}
    &\delta_{\cc_{0\zero}}: f \mapsto \zero \star f &\delta_{\cc_{1\zero}}: f \mapsto f \star \zero \\
    &\delta_{\cc_{0\un}}: f \mapsto \un \star f &\delta_{\cc_{1\un}}: f \mapsto f \star \un
\end{align*}

\begin{definition}[Semantics of $\ldd\texttt{-O}$]
\label{def:semantics}
Let $(\phi,n)$ be a $\lddo$ graph.
\begin{itemize}
    \item $(\blacksquare,0)$ denotes the constant function $\zero$ of arity zero;
    \item $(\square,0)$ denotes the constant function $\un$ of arity zero;
    \item Letters in $\Delta$ denote elementary operators; 
    \item Concatenation of letters denote composition of operators;
    \item The empty word $\varepsilon$ denotes the identity operator; 
    \item The symbol $\diamond$ denotes the Shannon operator; 
    \item The parameter $n$ is the arity of the function;
\end{itemize}
We denote by $\semb{(\phi,n)}$ the Boolean function represented by $(\phi,n)$.  
\end{definition}

A \emph{model} of $\ldd\texttt{-O}$ is an instantiation of $\Delta$ with some letters.
For instance, the letters $\uu$ and $\cc_{1\zero}$ are used to encode the unary operators $\delta_\uu$ and $\delta_{\cc_{1\zero}}$, respectively; we use similar letters for the remaining canalizing variables.
The simplest possible model has no letters: $\Delta$ is empty and $\Delta^{*}$ contains one word of length zero, namely $\varepsilon$, which semantically corresponds to the identity operator $\delta_\varepsilon: f \mapsto f$.
The $\ldd\texttt{-O}$ graphs obtained for $\Delta = \emptyset$, after merging isomorphic subgraphs, are known in the literature as Shannon Decision Diagrams, or \sdd{}; we thus term this model $\ldd\texttt{-S}$.

\begin{example}[Running Example]
    \label{ex:run}
    The $\ldd\texttt{-S}$ graph of the Boolean function
    \[
        (x_0,x_1,x_2,x_3) \mapsto x_1 \xor x_2 \xor (\lnot x_0 \land x_3)
    \]
    is depicted in Figure~\ref{fig:s} where dashed and solid edges point respectively to the left and right operands of $\diamond$ (recall that the operator $\star$ represented by $\diamond$ is not commutative).
    For clarity, the terminal nodes are not merged and are explicitly labeled by their respective constant functions.
\end{example}
\begin{figure*}[th]
	\centering
	\subcaptionbox{$\ldd\texttt{-S}$\label{fig:s}}{\includeframe{1}{0.23}{trim=0cm 0cm 16.5cm 0cm}{figures/fig1}}
	\subcaptionbox{$\lddoUCuz{}$\label{fig:uc}} {\includeframe{2}{0.23}{trim=0cm 0cm 16.5cm 0cm}{figures/fig1}}
	\subcaptionbox{$\ldd\texttt{-O-NUC}$\label{fig:nuc}} {\includeframe{3}{0.23}{trim=0cm 0cm 16.5cm 0cm}{figures/fig1}}
	\subcaptionbox{$\lddoNUCX$\label{fig:nucx}} {\includeframe{14}{0.23}{trim=0cm 0cm 16.5cm 0cm}{figures/fig1}}
	\caption{$\ldd\texttt{-O}$ graphs for $(x_0,x_1,x_2,x_3) \mapsto x_1 \xor x_2 \xor (\neg x_0 \land x_3)$.}\label{fig:RE_oldd}
\end{figure*}
\begin{remark}
    Unlike \bdd{} variants, the diamond nodes in $\lddo$ are not labeled with the integers denoting the indices (or positions) of variables.
    Such information can be fully retrieved from the length of the words labeling the edges and the nesting depth of diamond nodes.
\end{remark}

The canonicity of the $\ldd\texttt{-S}$ data structure is obvious: each Boolean function has a unique $\ldd\texttt{-S}$ representation and every $\ldd\texttt{-S}$ represents unambiguously a unique Boolean function.
In the next section, we detail how such canonicity is achieved for non-trivial models ($\Delta \neq \emptyset$).

\subsection{Syntactic Reduction and Canonicity}
Each letter in $\Delta$ comes with an introduction rule.
For instance, an intuitive introduction rule for the letter `$\uu$' could be:
\[
    \inference[\textnormal{\texttt{intro-}$\uu$}]{(\phi,n) \diamond (\phi,n)}{\xrightarrow{\uu} (\phi,n)} \enspace .
\]
We use $\xrightarrow{\overbrace{\uu. \dots .\uu}^{n \text{ times}}}(\square,0)$ to denote $(\square,n)$ 
and a similar shorthand notation is used for $\blacksquare$.
This in turn allows the following introduction rules for the letters $\cc_{1\zero}$ and $\cc_{1\un}$:
\[
    \inference[\textnormal{\texttt{intro-}$\cc_{1\zero}$}]{(\phi,n) \diamond (\blacksquare,n)}{\xrightarrow{\cc_{1\zero}} (\phi,n)},\ 
    \inference[\textnormal{\texttt{intro-}$\cc_{1\un}$}]{(\phi,n) \diamond (\square,n)}{\xrightarrow{\cc_{1\un}} (\phi,n)}  \ .
\]
The letters $\cc_{0\zero}$, $\cc_{0\un}$ could be introduced similarly. 
\[
    \inference[\textnormal{\texttt{intro-}$\cc_{0\zero}$}]{(\blacksquare,n) \diamond (\phi,n)}{\xrightarrow{\cc_{0\zero}} (\phi,n)}, \ 
    \inference[\textnormal{\texttt{intro-}$\cc_{0\un}$}]{(\square,n) \diamond (\phi,n)}{\xrightarrow{\cc_{0\un}} (\phi,n)}
    \  .
\]
Notice that, although natural, the above intro rules are not unique. 
We say that a graph is \emph{reduced} w.r.t. a fixed set of intro rules if it is a fixed point for those rules. 
Depending on $\Delta$, the same graph may have distinct reduced representations.
For instance, if $\Delta$ contains both $\uu$ and $\cc_{1\zero}$, then with respect to the rules above 
the graph $(\blacksquare,0) \diamond (\blacksquare,0)$ reduces to 
either $\xrightarrow{\cc_{1\zero}} (\blacksquare,0)$ or $\xrightarrow{\uu} (\blacksquare,0)$. 
A simple way to avoid such non-determinism would be to apply the intro rules w.r.t. a fixed order.  
We shall see next, however, that there is an interplay between such an order and the way one defines the intro rules for the graph to capture syntactically all the semantic occurrences of the involved elementary operators. 

We focus in the sequel on the model $\lddoUC$ containing all useless and canalizing letters. 
We reduce the $\diamond$ node by 
applying a constructor, $\nc{\diamond}$, defined in infix notation as follows  
\begin{lstlisting}
let $\phi_0$ $\nc{\diamond}$ $\phi_1$ =
    match $\phi_0 \diamond \phi_1$ with
    | $\phi \diamond \phi$ $\camlarrow$ $\xrightarrow{\uu}  \phi$
    | $\phi \diamond (\square,n)$ $\camlarrow$ $\xrightarrow{\cc_{1\un}}  \phi$
    | $\phi \diamond (\blacksquare,n)$ $\camlarrow$ $\xrightarrow{\cc_{1\zero}} \phi$
    | $(\square,n) \diamond \phi$ $\camlarrow$ $\xrightarrow{\cc_{0\un}} \phi$
    | $(\blacksquare,n) \diamond \phi$ $\camlarrow$ $\xrightarrow{\cc_{0\zero}} \phi$
    | $\_$ $\camlarrow$ $\phi_0\diamond\phi_1$
\end{lstlisting}
where the first letter gets introduced first (if possible), then the second etc. 
Using $\nc{\diamond}$, we define a reduction operator, $[\cdot]$, inductively on the structure of the graph:  
\begin{lstlisting}
let rec  $[\phi]$ =
  match $\phi$ with
  | $(\blacksquare,0)$ | $(\square,0)$ $\camlarrow$ $\phi$ 
  | $\xrightarrow{\ell} \phi'$ $\camlarrow$ let $\psi_0 \diamond \psi_1 = \texttt{elim-}\ell (\phi')$ in $[\psi_0]\ \nc{\diamond}\ [\psi_1]$
  | $\phi_0 \diamond \phi_1$ $\camlarrow$ $[\phi_0]\ \nc{\diamond}\ [\phi_1]$
\end{lstlisting}
where $\texttt{elim-}\ell$ eliminates the letter $\ell$ by reversing its introduction rule. For instance,  
\[
    \inference[\textnormal{\texttt{elim-}$\cc_{0\zero}$}]{\xrightarrow{\cc_{0\zero}} (\phi,n)}{(\blacksquare,n) \diamond (\phi,n)} \enspace .
\]
\begin{lemma}
    \label{lem:idem}
    The reduction operator $[\cdot]$ is idempotent, that is $[[\phi]]=[\phi]$ for any graph $\phi$. 
\end{lemma}
The proof is by induction on the structure of the graph and relies on the fact that the letters are always introduced in the same order.  
The purpose of the $\texttt{elim-}\ell$ procedure is to deconstruct already introduced letters before reintroducing them w.r.t. to the implicit order of $\nc{\diamond}$. 
For instance $[\xrightarrow{\cc_{1\zero}}(\blacksquare,0)]$ gives $\xrightarrow{\uu}(\blacksquare,0)$ since $\cc_{1\zero}$ is eliminated before introducing $\uu$. 

The relative canonicity of $\ldd\texttt{-O}$ models is proved by induction on the arity $n$. 
\begin{proposition}[Relative Canonicity of $\lddoUC$]
    \label{prop:canonicity}
    Every Boolean function has a unique, semantically equivalent, $\lddoUC$ graph $\phi$ reduced w.r.t. $[\cdot]$ (that is $[\phi]=\phi$).  
\end{proposition}
\begin{proof}
    Let $f$ denote a Boolean function of arity $n$. 
    Any Boolean function has a unique, semantically equivalent, $\ldds$ graph $\phi$ which is also a $\lddoUC$ graph. 
    The existence is a consequence of the fact that the procedure $[\cdot]$ is an idempotent function so $[\phi]$ exists and $[[\phi]] = [\phi]$. 

    It remains to prove uniqueness, that is there is no other distinct graph $\psi$ which is also semantically equivalent to $f$ and such that $[\psi] = \psi$. 
    The proof is by induction on the arity $n$.  
    Uniqueness trivially holds for $n=0$, as $(\blacksquare,0)$ and $(\square,0)$ are the unique fixed point for $[\cdot]$ with arity zero. 
    Suppose that uniqueness holds for arity $n$.

    {\bf Case 1.} Suppose that $[\phi]$ has the form $(\phi_1,n)\ \diamond (\phi_2,n)$. 
    If $\psi$ has the form $\xrightarrow{\ell} (\psi',n)$ then by eliminating $\ell$ we get $(\psi'_1,n) \diamond (\psi'_2,n)$ which is semantically equivalent to $(\phi_1,n) \diamond (\phi_2,n)$ thus $\semb{\phi_1} \equiv \semb{\psi_1'}$ and $\semb{\phi_2} \equiv \semb{\psi_2'}$ and the induction hypothesis implies that both pairs should be syntactically equal but then this means that $(\phi_1,n)\ \diamond (\phi_2,n)$ is not reduced contradicting the hypothesis. 
    Now if $\psi$ has the form $(\psi'_1,n) \diamond (\psi'_2,n)$ then a similar reasoning leads to the fact that $\psi$ and $[\phi]$ must be syntactically equal.  

    {\bf Case 2.} Suppose that $[\phi]$ has the form $\xrightarrow{\ell} (\phi',n)$. 
    Like the first case, if $\psi$ has the form $(\psi'_1,n)\diamond(\psi'_2,n)$ then by eliminating $\ell$ we get a contradiction. 
    Now suppose $\psi$ has the form $\xrightarrow{\ell'} (\psi',n)$. By eliminating both letters and using the semantic equivalence together with the induction hypothesis we get that $\ell = \ell'$ as they were both introduced using $\nc{\diamond}$ applied to the exact same graphs. 
    Thus $\psi = [\phi]$. 
\end{proof}

The relative dependence to $[\cdot]$ highlights a subtle dependence between the way $[\cdot]$ is defined and the intro rules it uses. 
Let's see what this precisely means on a concrete example. 
Consider the graph $\phi := (\square \diamond \square) \diamond (\blacksquare \diamond \blacksquare)$. 
Then $[\phi]$ gives $\phi_1 := \xrightarrow{\cc_{1\zero}.\uu}\square$. 
Now, suppose we swap the first and fifth matchings of $\nc{\diamond}$ so that all canalizing letters get introduced before $\uu$. 
Then $\phi$ gets instead reduced to $\phi_2 := (\xrightarrow{\cc_{1\un}}\square) \diamond (\xrightarrow{\cc_{0\zero}}\blacksquare)$. 
If, however, we swap the third and fourth matchings, one gets $\phi_3 := \xrightarrow{\cc_{0\un}.\uu}\blacksquare$. 

Although $\phi_1$, $\phi_2$ and $\phi_3$ are semantically equivalent and unique w.r.t. their respective reduction operators, $\phi_1$ and $\phi_3$ are better than $\phi_2$. 
Indeed, not only they are more concise, they also capture all (semantically) useless and canalizing variables of the Boolean function they represent, albeit each uses a different letter to encode the canalizing variable. 

To ensure that the reduction operator captures, at the syntactic level, all occurrences of the elementary functors involved in the Boolean function, it must respect the implicit dependence between the letters in the considered intro rules.  
For the $\lddoUC$ model, the intro rules we used for canalizing letters have the letter $\uu$ involved in their very definitions 
since they expect either $(\square,n)$ or $(\blacksquare,n)$ to occur on one side of $\diamond$. 
Thus exhibiting $\uu$ letters before any other canalizing letter (like $[\cdot]$ does)  permits to canonically represent constant Boolean functions of any arity as $(\square,n)$ or $(\blacksquare,n)$, thereby syntactically capturing any potential occurrence of (semantic) canalizing variables via the intro rules we defined. 
We can now state the main result of this section: the canonicity of the $\lddoUC$ model. 
\begin{theorem}[Canonicity of $\lddoUC$]
    Every Boolean function $f$ has a unique, semantically equivalent, $\lddoUC$ graph $\phi$ satisfying the following property: the variable introduced by each $\diamond$ node of $\phi$ is neither a useless nor a canalizing variable. 
\end{theorem}
In other words, the theorem says that all useless and canalizing variables are captured in the graph $\phi$ by one of the letters. 
We next explain how negation is supported in our framework.

\section{Negation Operator}\label{sec:negation}
In this section we enrich our data structure with arity-preserving operators, such as negation. 
This extension is novel: none of the recent \bdd{} variants~\cite{CBDD,ESRBDD-TACAS2019} that partially use canalizing variables support negation. 
As we shall see in the sequel, our functional framework gives a clear insight into why this is difficult for those variants. 

Arity-preserving operators are not elementary, since they do not increase the arity of their operands.
We exemplify the main steps through the concrete example of the standard output negation operator on Boolean functions $\delta_\lnot: \F{n} \to \F{n}$, defined by
\(
    (\delta_\lnot f): (x_0,\dotsc,x_{n-1}) \mapsto \lnot f(x_0,\dotsc,x_{n-1})
\). 
Syntactically, the words labeling the edges will now include a new letter `$\bullet$' that encodes the negation operator. 

The letter `$\bullet$' is introduced by normalizing the representation of the square node $(\square,0)$:~\footnote{Choosing $(\square,0)$ over $(\blacksquare,0)$ is arbitrary and has no impact on the treatment that follows.}
\[
    \inference[\textnormal{\texttt{intro-$\bullet$}}]{(\square,0)}{\xrightarrow{\bullet} (\blacksquare,0)} \enspace .
\]
Arity-preserving operators in general, and the negation operator in particular, interact with both the elementary operators and the Shannon operator.
Those interactions have to be taken into account for the words labeling the edges of the graph to be canonical.
There are two fundamental aspects of such interactions.
The first is the distributivity with the Shannon operator.
The second is the commutativity with elementary operators: when and how do arity-preserving operators
commute with other operators?

The negation operator distributes over the Shannon operator, which makes it possible to propagate it upward. 
The related normalization rule given below is similar to the \bdd{} variants supporting complement edges.\footnote{One may arbitrarily choose the dual reduction rule that normalizes the negation to the left edge.}
\[
    \inference[\textnormal{\texttt{norm-Shannon}}]{(\xrightarrow{\bullet} (\phi_1,n)) \diamond (\phi_2,n)}{\xrightarrow{\bullet} ((\phi_1,n) \diamond (\xrightarrow{\bullet} (\phi_2,n)))}
    \enspace .
\]

Commuting `$\bullet$' with the letters encoding the elementary operators allows to syntactically exploit the involution property of negation: 
\[
        \inference[\textnormal{\texttt{norm-involution}}]{\xrightarrow{\bullet} (\xrightarrow{\bullet} (\phi,n))}{(\phi,n)} \enspace .
\]
However, commutation with elementary operators is not always possible.
For instance the following two functions are not semantically equivalent (the standard functional notations on the   right are given explicitly for convenience). 
\begin{align*}
&(\delta_{\cc_{1\zero}}\circ \delta_\neg) \zero = \un \star \zero          &\qquad (x &\mapsto \neg x) \\
&(\delta_\neg\circ \delta_{\cc_{1\zero}}) \zero = \neg (\zero \star \zero)        &\qquad (x &\mapsto 1)
\end{align*}

To overcome this issue, we adopt in this work a weaker notion of commutativity.
\begin{definition}[Quasi-commutativity]
    We say that an arity-preserving letter, $a$, \emph{quasi-commutes} with an elementary letter, $\ell$, if 
there exists an elementary letter $\ell_a$ (possibly different from $\ell$) such that
the graphs $\ell.a (\phi,n)$ and $a.\ell_a (\phi,n)$ are semantically equivalent for all $(\phi,n)$. 
(As the notation suggests, the letter $\ell_a$ depends in general on $a$.)
We will say that an alphabet $\Delta$ \emph{quasi-commutes} with $a$, or is $a-$stable, if for all $\ell \in \Delta$, $\ell_a \in \Delta$.
\end{definition}
Semantically, quasi-commutation of letters encodes the quasi-commutation of elementary operators and arity\-/preserving operators.
For each arity\-/preserving operator $a$, the following normalization rule, defined for each elementary operator $\ell$, exploits quasi-commutativity to expose arity\-/preserving operators before elementary operators. 
\[
\inference[\textnormal{\texttt{norm-commutation-}$(\ell,a)$}]{\xrightarrow{\ell.a} (\phi,n)}{\xrightarrow{a.\ell_a} (\phi,n)} \enspace .
\]
We stress the fact that $\ell$ must be an elementary letter. Hence, one cannot apply \texttt{norm-commutation} on $\xrightarrow{\bullet.\bullet} (\phi,n)$. 
Quasi-commutativity naturally extends to words: a word $w$ quasi-commutes with an arity-preserving letter $a$ if and only if all its letters quasi-commute with $a$. 

The negation, being involutive, is a special operator. 
Indeed for any letter $\ell$, ${(\ell_\bullet)}_\bullet = \ell$. 
Thus, one can make any alphabet $\Delta$ $\bullet$-stable by saturating it with the elements $\ell_\bullet$ for each $\ell \in \Delta$.
As defined, the involution and commutation rules of the negation expose the $\bullet$ letter at the beginning of each word.
The negation commutes with useless variables (i.e., $\uu.\bullet$' as `$\bullet.\uu$). 
However, it only quasi-commutes with all the other canalizing variables:  
\begin{align*}
    \label{eq:quasicommutation}
\cc_{0\zero}.\bullet &= \bullet.\cc_{0\un}    & \cc_{0\un}.\bullet &= \bullet.\cc_{0\zero} & ({\cc_{0\zero}}_\bullet = \cc_{0\un}) \\
\cc_{1\zero}.\bullet &= \bullet.\cc_{1\un}    & \cc_{1\un}.\bullet &= \bullet.\cc_{1\zero} & ({\cc_{1\zero}}_\bullet = \cc_{1\un})
\end{align*}

Quasi-commutation explains the difficulty in adding (and normalizing) negation in some \bdd{} variants like \zdd{}.
Typically, a model that has the letter $\cc_{0\un}$ without the letter $({\cc_{0\un}})_\bullet = \cc_{0\zero}$ cannot properly support the encoding and propagation of the negation over the data structure: the normalization becomes overly cumbersome while not offering any clear advantage.

We term the model $\lddoUC$ enriched with `$\bullet$' $\lddoNUC$. 
The reduction operator $[\cdot]$ defined in the previous section is extended to account for the introduction and normalizing rules of $\bullet$ while respecting the implicit dependencies of the introduction rules. 
The detailed procedure is given in Appendix~\ref{ax:norm-negation}. 
We prove also that $\lddoNUC$ is canonical as part of the formalization of $\lddo$ data structure in the Coq proof assistant.   
\begin{theorem}[Canonicity of $\lddoNUC$]
    Every Boolean function $f$ has a unique, semantically equivalent, $\lddoNUC$ graph $\phi$ satisfying the following property: the variable introduced by each $\diamond$ node of $\phi$ is neither a useless nor a canalizing variable.  
    Moreover the canonical $\lddoNUC$ graph of $\lnot f$ differs from $\phi$ only by the letter $\bullet$ at the upper-most edge of the graph.  
\end{theorem}

Figure~\ref{fig:nuc} shows the canonical $\lddoNUC$ graph of the Boolean function of Example~\ref{ex:run}. 




\section{Functional Classification\label{sec:dd-unify}}
The functional point of view suggests a natural way to classify and enumerate a vast class of ordered models.
It turns out that a large body of already existing \bdd{} variants can be seen, through the lenses of $\ldd$, as special ordered models.

In this section, we start by exhaustively enumerating all arity-preserving functors that act on $f \in \F{n}$ by transforming its output $f(x_0,\dotsc,x_{n-1})$.~\footnote{Other arity-preserving functors \cite{RolfVariantDual,BurchLong1992} are possible and equally interesting. Those are however out of the scope for this paper, but planned in the near future.}


An arity-preserving operator that acts on $f \in \F{n}$ by transforming its output $f(x_0,\dotsc,x_{n-1})$ necessarily has the form:
\[
    (\delta_p f): (x_0,\dotsc,x_{n-1}) \mapsto p(f(x_0,\dotsc,x_{n-1}))
\]
where $p$ is a Boolean function in $\F{1}$.
There are $2^{2^1} (= 4)$ possibilities for $p$: the two constant functions $\zero$ and $\un$ of arity $1$,
the identity function $\imath:x \mapsto x$, and the negation $\lnot:x \mapsto \lnot x$.
When $p$ is a constant function, $\delta_p$ is not injective and is therefore of
little interest when it comes to canonical representations.
When $p$ is the identity $\imath$, then $\delta_p$ is the identity operator $\varepsilon$.
Thus, one recovers the model $\ls$.
Finally, when $p$ is the negation, $\delta_p$ corresponds to $\delta_\lnot$ and the so obtained model, termed $\lsn$, enriches $\ls$ with the negation operator.

In a similar fashion, we enumerate elementary operators that act on $f \in \F{n}$ by combining its output $f(x_1,\dotsc,x_{n})$ with a fresh variable $x_0$. Such operator necessarily has the form
\[
    (\delta_p f): (x_0,\dotsc,x_{n}) \mapsto p(x_0, f(x_1,\dotsc,x_{n})) \enspace ,
\]
where the parameter $p$ is now an element of $\F{2}$.
Let $\pi_1 \in \F{(n+1)}$ denote the projection operator returning the first input.
We can enumerate the $2^{2^2} (= 16)$ possibilities for $p$ by combining two Boolean functions of arity $1$ using the Shannon operator.
Constant operators (2 cases), as well as the operators involving the projection $\pi_1$ (2 cases), are not injective.
Therefore, they cannot be used for canonical representations.
Except for two injective operators, the rest can be fully captured by one of the elementary operators we have already introduced (related to useless or canalizing variables---see Definitions~\ref{def:uselessvar} and~\ref{def:canalizingvar}), possibly combined with the negation operator.
The two remaining injective elementary operators reveal a new kind of variable which is neither useless nor canalizing. 
\begin{definition}[Xor Variable]
\label{def:xorvar}
Let $\delta_\xx$ denote the following elementary operator:
\[
    \delta_\xx: \F{n} \to \F{(n+1)}, \quad f \mapsto f \star \delta_\neg (f) \enspace .
\]
A variable introduced with $\delta_\xx$ is called a \emph{xor variable}.
\end{definition}
Although one can define a model solely with $\delta_{\xx}$, without supporting the negation as an extra operator,
such a model would not necessarily be useful, as the reduction of the xor-variables cannot be performed in constant time over normalized subgraphs.
Thus, such operator is much more relevant when the negation is properly supported (i.e., propagated and normalized as discussed in Section~\ref{sec:negation}).
This enumeration suggests that a model defined using
\[
\{\uu,\xx,\cc_{0\zero},\cc_{1\zero},\cc_{0\un},\cc_{1\un}\} \cup \{\bullet\} \enspace ,
\]
where the letter `$\xx$' encodes $\delta_\xx$, would be the most expressive, $\bullet$-stable, model of the considered class as it has all the elementary operators plus the negation.
The next section is entirely devoted to this model, termed $\lddoNUCX$.
\begin{table}[ht]
    \centering
\scriptsize
\caption{\label{tab:variant-model-o} Existing \bdd{} variants with their corresponding $\ldd$ ordered models. }
\renewcommand{\arraystretch}{1.3}
\begin{tabular}{|r|l|r@{\hskip 2pt}c@{\hskip 2pt}l|}
\hline
\bfseries Variant & \bfseries Model &  \multicolumn{3}{r|}{\bfseries Alphabet} \\ 
\hline
SDD    &   $\ls$	& $\emptyset$ & $\cup$ & $\emptyset$ \\ 
SDD+N  &   $\lsn$	& $\emptyset $ & $\cup$ & $ \{\bullet\}$ \\ 
\hline 
ROBDD~\cite{Bryant1986}  &  $\lddoU{}$	& $\{\uu\} $ & $\cup$ & $ \emptyset$ \\ 
ROBDD+N~\cite{Somenzi1999}  &  $\lddoNU{}$	& $\{ \uu \} $ & $\cup$ & $ \{\bullet\}$ \\
ZDD~\cite{ZDD}   & $\lddoCuz$	& $\{\cc_{1\zero}\} $ & $\cup$ & $ \emptyset$ \\
ChainDD~\cite{CBDD} & $\lddoUCuz{}$ & $\{\uu, \cc_{1\zero}\}$ & $\cup$ &  $\emptyset$ \\
ChainDD+N~\cite{CBDD+N} & $\lddoNUCuzuu{}$ & $\{\uu, \cc_{1\zero}, \cc_{1\un}\}$ & $\cup$ &  $\{\bullet\}$ \\
ESR$\text{-L}_0$~\cite{ESRBDD-TACAS2019} & $\lddoUCuz{}$ & $\{\uu, \cc_{1\zero}\}$ & $\cup$ &  $\emptyset$ \\
ESR~\cite{ESRBDD-TACAS2019}  &   $\lddoUCz{}$	& $\{\uu, \cc_{0\zero}, \cc_{1\zero}\} $ & $\cup$ & $ \emptyset$ \\ 
\hline
\texttt{DAGaml-O-NUCX}       & $\lddoNUCX{}$    & $\{\uu,\xx,\cc_{0\zero},\cc_{1\zero},\cc_{0\un},\cc_{1\un}\} $ & $\cup$ & $\{\bullet\}$ \\
\hline
\end{tabular}
\end{table}

Table~\ref{tab:variant-model-o} summarizes some $\ldd\texttt{-O}$ models and their related variants (or implementations).
Observe that the \texttt{TBDD}~\cite{TaggedBDD} variant is not an ordered model as it uses a syntactic negation that does not correspond to the (functional) standard negation.\footnote{However, this variant can be captured as a model enriched with a specific operator that captures the semantics of its syntactic negation.}
The last column of the table gives the alphabet of each model split into two subsets: elementary letters (left) and arity-preserving letters (right) where the expressiveness increases top down. 

The ordered models introduced so far form in fact a complete lattice.
Its (partial) Hasse diagram is depicted in Figure~\ref{fig:lattice-ordered-models}.
The least upper bound (resp. greatest lower bound) of two models is the model induced by the union (resp. intersection) of their $\Delta$ alphabets.
The first layer of the diagram has exactly $13$ elements: one per elementary operator ($\{\uu, \xx, \cc_{0\zero}, \cc_{0\un}, \cc_{1\zero}, \cc_{1\un}\}$), one per negated elementary operator (not necessarily $\bullet$-stable), and one involving the negation only.
A model $M_2$ is more expressive than another model $M_1$ if and only if there is a path from $M_1$ to $M_2$ in the diagram.
In particular, if there is no path between $M_1$ and $M_2$, then they are incomparable.
This relation translates immediately to the number of nodes: the more expressive the model is, the smaller its number of nodes.
Observe for instance, the graphs of Figure~\ref{fig:RE_oldd}, where the number of nodes is decreasing from left to right.

\begin{figure}
	\centering
	\begin{tikzpicture}[transform shape, scale=0.75]
	\node (nucx) {$\ldd$\texttt{-O-NUCX}};
	\node (robddn) [below left of=nucx] {$\lddoNU{}$};
	\node (esr) [below right of=nucx] {$\lddoUCz$};
	\node (chain) [below of=esr] {$\lddoUCuz$};
	\node (zdd) [below of=chain] {$\lddoCuz$};
	\node (robdd) [left of=zdd] {$\ldd$\texttt{-O-U}};
	\node (sdd) [below left of=zdd] {$\ls{}$};
	\node (sddn) [left of=robdd] {$\lsn{}$};
	\draw[->,>=stealth] (robdd) to node [swap] {} (robddn);
	\draw[->,>=stealth] (zdd) to node [swap] {} (chain);
	\draw[->,>=stealth] (chain) to node [swap] {} (esr);
	\draw[->,>=stealth] (robddn) to node [swap] {} (nucx);
	\draw[->,>=stealth] (esr) to node [swap] {} (nucx);
	\draw[->,>=stealth] (sdd) to node [swap] {} (robdd);
	\draw[->,>=stealth] (sdd) to node [swap] {} (zdd);
	\draw[->,>=stealth] (robdd) to node [swap] {} (chain);
	\draw[->,>=stealth] (sdd) to node [swap] {} (sddn);
	\draw[->,>=stealth] (sddn) to node [swap] {} (robddn);
	\end{tikzpicture}
	\caption{Some $\ldd$ ordered models organized in a lattice (cf. Table~\ref{tab:variant-model-o} for their corresponding variants).}
	\label{fig:lattice-ordered-models}
\end{figure}
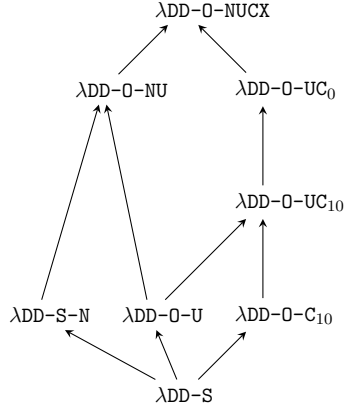

It becomes apparent that \chainB{}, \chainZ{}~\cite{CBDD}, and \esrlz~\cite{ESRBDD-TACAS2019} are three possible 
implementations of the $\lddoUCuz{}$ model.
The main difference between these variants resides in their respective, carefully designed, choices of encoding
the involved elementary operators either as labels or special nodes.
\esr{}, for instance, encodes canalizing variables as nodes with one child.
From a functional point of view, however, they are indistinguishable. 
To the best of our knowledge, such observation has never been made in the literature where only a performance comparison prevailed. 
Notice also that these three variants, as well as their related model, are not $\bullet$-stable, hindering the support of constant-time negation. 
This gives a clear insight into a fundamental limitation of these models. 
\begin{remark}
    In \cite[Section 9]{CBDD+N}, the author mentions an extension of \chainB{}~\cite{CBDD} with complement edges. 
    The section is fairly short and doesn't explain how this was done while preserving canonicity. 
    We suspect that the author used both $\cc_{1\zero}$ and $\cc_{1\un}$ to make the variant stable by negation. 
    Recall that these two letters quasi-commute with negation. 
    Let's also stress the fact that the two letter used in \esr{}~\cite{ESRBDD-TACAS2019} do not quasi-commute making it difficult to support complement edges with their natural functional semantics. 
\end{remark}

\section{$\lddoNUCX$}\label{sec:onucx}

We discuss in this section the most expressive canonical model, called $\lddoNUCX$, that supports all elementary operators of the class $\lddo$, together with negation:
\(
\Delta:= \{\uu,\xx,\cc_{0\zero},\cc_{1\zero},\cc_{0\un},\cc_{1\un}\} \cup \{\bullet\}
\). 
We use the following introduction rule for the letter $\xx$ (encoding xor-variables).
\[
    \inference[\textnormal{\texttt{intro-}$\xx$}]{(\phi,n) \diamond \bigl( \xrightarrow{\bullet}(\phi,n)\bigr)}{\xrightarrow{\xx} (\phi,n)}  \enspace .
\]
The letter `$\xx$' commutes with `$\bullet$', that is the word `$\xx.\bullet$' can be rewritten as `$\bullet.\xx$'.
Thus, $\Delta$ is $\bullet$-stable. 

The extension of the reduction operator $[\cdot]$ is much more involved and is part of the formalization of the $\ldd$ data structure where we also proved the canonicity of the model. 

\begin{theorem}[Canonicity of $\lddoNUCX$]
\label{thm:canon}
    Every Boolean function $f$ has a unique, semantically equivalent, $\lddoNUCX$ graph $\phi$ satisfying the following property: the variable introduced by each $\diamond$ node of $\phi$ is neither a useless nor a canalizing variable nor a xor-variable.
    Moreover the canonical $\lddoNUCX$ graph of $\lnot f$ differs from $\phi$ only by the letter $\bullet$ at the upper-most edge of the graph.  
\end{theorem}

To contrast this new model with the most recent \bdd{} variants, namely ChainDD~\cite{CBDD} and ESR$\text{-L}_0$~\cite{ESRBDD-TACAS2019}, the $\lddoUCuz{}$ graph of the Boolean function of Example~\ref{ex:run} is depicted in Figure~\ref{fig:uc} whereas its $\lddoNUCX$ graph is given in Figure~\ref{fig:nucx}. The latter is clearly more concise as it has only one diamond node.
Negating such graph in Chain-DD or ESR would require reconstructing the entire diagram as both don't support negation. 
Negating it in $\lddoNUCX$ amounts to simply add a `$\bullet$' in the label of its uppermost edge. 

%

\subsection{Normalizing Logical Connectives}
\label{ax:normalizing}
In practice, normalized $\lddoNUCX$ graphs are built by normalizing the logical connectives.
We detail below the procedure `\texttt{andb}' that computes the conjunction of two normalized graphs.
The algorithm is easily adaptable to the other operations.
Computations are performed as usual, by recursively pushing the logical operator down to the child graphs of diamond nodes.
For clarity, we denote a $\ldd$ graph $(\phi,n)$ simply by $\phi$, omitting arity $n$ whenever unnecessary.
We consider that the arity can be always extracted from $\phi$ by calling $\operatorname{arity}(\phi)$.
\begin{lstlisting}
let rec andb $\phi$ $\psi$ =
  if $\phi = \psi$ then $\phi$
  else if $\phi$ = $\xrightarrow{\bullet}\psi$ then $(\blacksquare, \operatorname{arity}(\phi))$
  else match $\phi$ with
    | $(\blacksquare, n)$ $\camlarrow$ $(\blacksquare, n)$
    | $\xrightarrow{\bullet}(\blacksquare, n)$ $\camlarrow$ $\psi$
    | _ $\camlarrow$ (match $\psi$ with
      | $(\blacksquare, n)$ $\camlarrow$ $(\blacksquare, n)$
      | $\xrightarrow{\bullet}(\blacksquare, n)$ $\camlarrow$ $\phi$
      | _ $\camlarrow$ (
        let $\phi_0$ = cofactor $0$ $\phi$
        and $\phi_1$ = cofactor $1$ $\phi$
        and $\psi_0$ = cofactor $0$ $\psi$
        and $\psi_1$ = cofactor $1$ $\psi$ in
        (andb $\phi_0$ $\psi_0$) $\ncn{\diamond}$ (andb $\phi_1$ $\psi_1$) ))
\end{lstlisting}
%
%
`\texttt{cofactor} $v_0$ $\psi$' is defined below. Intuitively, it returns the graph of the Boolean function $\psi$ when its first argument is set to $v_0$.
%
\begin{lstlisting}
let cofactor $v_0$ $\phi$ =
  match $\phi$ with
  | $\xrightarrow{\bullet}\phi'$ $\camlarrow$ $\xrightarrow{\bullet}$(cofactorElem $v_0$ $\phi'$)
  | _ $\camlarrow$ (cofactorElem $v_0$ $\phi$)
\end{lstlisting}

\begin{lstlisting}
let cofactorElem $v_0$ $\phi$ =
  match $\phi$ with
  | $(\blacksquare, n)$ $\camlarrow$ $(\blacksquare, n-1)$
  | $\xrightarrow{l} \phi'$ $\camlarrow$ match $l$ with
    | $\uu$ $\camlarrow$ $\phi'$
    | $\cc_{b\tt}$ $\camlarrow$ (if $v_0 = b$ then $(\tt, n-1)$ else $\phi'$)
    | $\xx$ $\camlarrow$ (if $v_0$ then $\xrightarrow{\bullet}\phi'$ else $\phi'$)
\end{lstlisting}
%

We discuss below the time complexity of \lstinline|andb| $\phi$ $\psi$. 
Let $\abs{\phi}$ denote the number of diamond nodes of the normalized graph $\phi$, and
let $\abs{\phi}_\texttt{S}$ denote the number of diamond nodes of its equivalent $\ldds$ graph (obtained by eliminating all letters and using diamond nodes instead). 
The algorithm \lstinline|andb| can be applied almost identically to $\ldds$ graphs except for a
minor edit to account for the terminal node $(\square,0)$.
The algorithm performs a simple structural induction on its inputs.
Assuming memoization, the number of recursive calls is bounded by $\mathcal{O}(\abs{\phi}_\texttt{S}\times\abs{\psi}_\texttt{S})$.
The time complexity of a single recursive call is $\mathcal{O}(1)$ as the complexity of \lstinline|cofactor| is constant time (finite branching and no loops).
Thus, the overall time complexity is $\mathcal{O}(\abs{\phi}_\texttt{S}\times\abs{\psi}_\texttt{S})$.

We have $\abs{\phi}_\texttt{S}$ is equal to $\abs{\phi}$ plus the total size (or length) of all words in $\phi$ (each letter is a reduced diamond node).
The size of any word in $\phi$ is bounded by $n$, the total number of variables.
The total number of edges in $\phi$ is $1+2\abs{\phi}$.
Thus the total size of all words is bounded by $(1+2\abs{\phi})n$ and, therefore
\[
    \abs{\phi}_\texttt{S} \leq n + \abs{\phi} + 2 n \abs{\phi} = \mathcal{O}(n\times \abs{\phi}) \enspace .
\]
This leads to an overall time complexity bounded by $\mathcal{O}(n^2\times\abs{\phi}\times\abs{\psi})$.

\subsection{Complexity of Common Queries}\label{sec:complex}
Common queries on \robdd{}  (e.g., TAUTOLOGY, EQUIVALENCE, SAT, AnySAT, AllSAT, \#SAT) have polynomial time complexity on the size of the $\lddoNUCX$ graphs.

The simplest way to check for EQUIVALENCE($(\phi, n)$, $(\psi, n)$) is by hash-consing both $\phi$ and $\psi$. 
Hash-consing has a linear time complexity in the size of both graphs leading to $\mathcal{O}(n \times (\abs{\phi} + \abs{\psi}))$.

For SAT, it suffices to check whether the graph is $(\blacksquare,n)$.
In the worst case, the entire word (of size at most $n$) has to be checked, leading to a time complexity of $\mathcal{O}(n)$.
Likewise for TAUTOLOGY.

One can compute \#SAT inductively on the structure of a $\lddo$ graph:
\begin{lstlisting}
let rec count $(\phi, n)$ =
  match $(\phi, n)$ with
  | $(\blacksquare, n)$ $\camlarrow$ $0$
  | $(\square, n)$ $\camlarrow$ $2^n$
  | $\xrightarrow{\bullet}(\phi', n')$ $\camlarrow$ $2^n-\operatorname{count}(\phi', n)$
  | $\xrightarrow{\uu}(\phi', n')$ $\camlarrow$ $2\times\operatorname{count}(\phi', n)$
  | $\xrightarrow{\xx}(\phi', n')$ $\camlarrow$ $2^{n'}$
  | $\xrightarrow{\cc_{0\zero}}(\phi', n')$ | $\xrightarrow{\cc_{1\zero}}(\phi', n')$ $\camlarrow$ $\operatorname{count}(\phi', n)$
  | $\xrightarrow{\cc_{0\un}}(\phi', n')$ | $\xrightarrow{\cc_{1\un}}(\phi', n')$ $\camlarrow$ $2^{n'}+\operatorname{count}(\phi', n)$
  | $(\phi_0, n')\diamond(\phi_1, n')$ $\camlarrow$ $\operatorname{count}(\phi_0, n') + \operatorname{count}(\phi_1, n')$
\end{lstlisting}
Using memoization, it can thus be computed in $\mathcal{O}(n\times\abs{\phi})$.
Computing AnySat or AllSat can be performed similarly by induction over the structure.
In particular, AnySat can be computed in $\mathcal{O}(n)$, and AllSat in $\mathcal{O}(n\times \#SAT(\phi))$.

\section{Compression Factors}
\label{sec:size}
We compare the size of the different $\lddo$ models presented so far. 
Let $f$ denote a Boolean function. 
We denote by $N_{\texttt{A}}(f)$ the total number of nodes in the diagram that represents $f$ with respect to model $\texttt{A}$. 
For clarity, we restrict our attention to models with no complement edges (that is models without the $\delta_\neg$ functor). 
Knuth~\cite{knuth4A} showed that~\footnote{where $\ldds$ are referred to as quasi-BDD.}
\[
N_{\lddoU}(f) \leq N_{\ldds}(f) \leq \frac{n+1}{2} (N_{\lddoU}(f) + 1) 
\]
and similarly for $N_{\lddoCuz}(f)$ (which is equivalent to \zdd). 
In fact the same inequality holds for any (non-negated) model of the first layer of the Hasse diagram (see Figure~\ref{fig:lattice-ordered-models}).  
This in particular allows to compare the size of any two incomparable models. 
For instance, for \zdd{} and \bdd{} one gets 
\[
N_{\lddoCuz}(f) \leq \frac{n+1}{2} (N_{\lddoU}(f) + 1) \quad \text{and} \quad N_{\lddoU}(f) \leq \frac{n+1}{2} (N_{\lddoCuz}(f) + 1) \enspace .
\]
These inequalities bound the potential gain from using one model over the other and shows that such gain is linear in the number of variables (which may be considerable when $n$ is big). 
The following immediate generalization holds. 
\begin{theorem}\label{thm:size}
If the model $\texttt{B}$ is more expressive than model $\texttt{A}$ with respect to the Hasse diagram of 
Figure~\ref{fig:lattice-ordered-models}, then 
\[
N_{\texttt{B}}(f) \leq N_{\texttt{A}}(f) \leq \frac{n+1}{2} (N_{\texttt{B}}(f) + 1) \enspace .
\]
Consequently if models $\texttt{B}_1$ and $\texttt{B}_2$ are two incomparable models that are both more expressive than $\texttt{A}$, then 
\[
N_{\texttt{B}_1}(f) \leq \frac{n+1}{2} (N_{\texttt{B}_2}(f) + 1) \quad \text{and} \quad N_{\texttt{B}_2}(f) \leq \frac{n+1}{2}              (N_{\texttt{B}_1}(f) + 1) \enspace .
\]
\end{theorem}
\begin{proof}
The reasoning to show these inequalities is essentially the same as in~\cite{knuth4A}. 
To count the size of the diagram, one counts the nodes at each level starting from level $0$ for the root node all the way down to level $n$ with the terminal nodes. 
For a level $k$ of the diagram, let $a_k$ (resp. $b_k$) denote its total number of nodes for model $\texttt{A}$ (resp. $\texttt{B}$). 
One shows that 
\[
    a_k \leq 1 + \underbrace{b_0 + \dotsb + b_{k-1}}_{\mu_k} \quad \text{and} \quad a_k \leq b_k + \dotsb + b_{n} \enspace. 
\]
The first inequality holds because, as soon as a variable is typed, its related node is removed from the diagram making its in-degree branches to go necessarily past level $k$. 
Since there are $\mu_k$ nodes above level $k$, the total number of branches leaving those nodes is $2 \mu_k$, among which $\mu_k - 1$ are used (to connect the $\mu_k$ nodes) thus leaving $1 + \mu_k$ branches that are connected to nodes below level $k$. 
Each branch corresponds necessarily to a function of arity $n-k$ and $a_k$ is precisely the number of distinct (sub)functions of arity $n-k$. Thus $a_k \leq 1 + \mu_k$.  

For the second inequality, every (sub)function at level $k$ corresponds necessarily to one function among the $b_k + \dotsb + b_{n}$ functions up to extracting its typed variables (if any). Summing up the two inequalities, one gets 
\[
    N_{\texttt{A}}(f) \leq \frac{n+1}{2} (N_{\texttt{B}}(f) + 1) \enspace .
\]
For the remaining inequality, since $\texttt{B}$ is more expressive than $\texttt{A}$, it is obvious that $N_{\texttt{B}}(f) \leq N_{\texttt{A}}(f)$ because some nodes get eventually removed from the diagram of $f$ in $\texttt{A}$. 
The inequalities comparing the sizes of $\texttt{B}_1$ and $\texttt{B}_2$ are obtained by transitivity via $\texttt{A}$.  
\end{proof}
Theorem~\ref{thm:size} gives a good estimate of the potential gain one may get when using more expressive models. 
It is interesting to note that, even for the most expressive model, the gain can only be linear in $n$ at most. 

One shows also that when negation is no longer supported, the size of the diagram can double: 
\[
    N_{\texttt{A}+\bullet}(f) \leq N_{\texttt{A}}(f) \leq 2 N_{\texttt{A}+\bullet}(f),  
\]
where $\texttt{A}+\bullet$ denotes the model $\texttt{A}$ extended with the negation functor $\delta_\neg$.  

The main drawback of such analysis is that it doesn't account for the size of the labels (or any other artefact) used to encode chains of typed variables. These choices have naturally an impact on the potential gain (see for instance the bounds reported in~\cite{CBDD} where the author used special nodes to encode typed variables).  

When labels are used, one can bound the overall size of all the labels by $(2N+1)n$ for a diagram with $N$ nodes representing a function defined over $n$ variables. Indeed each edge of the diagram has at most $n$ type symbols. 
The overall gain in the number of nodes may be in general compensated by the overhead induced by using labels. 
These practical considerations are however application and implementation dependent and are not discussed in this paper.  

\section{Related Work}\label{sec:relwork}

The $\ldd$-based classification differs from Darwiche's work~\cite{DarwicheClassification} that
uses relative compactness and absolute worst-time complexity of standard queries to classify representations of Boolean functions. 
Firstly, $\ldd$ is more fine grained.
With respect to Darwiche's classification, many important variants like \robdd{}, \robddn{}, \zdd{}, \chainD{}, \tbdd{} and \esr{} are mostly indistinguishable from the vanilla variant \sdd{} since they all fall in the same class.
Indeed, all $\ldd\texttt{-O}$ models we presented in this work handle several queries (e.g., EQUIVALENCE, SAT, \#SAT) in polytime, very much like \robdd{}. 
Secondly, unlike Darwiche's classification, the $\ldd$ framework focuses primarily on getting more (functionally) expressive canonical models in a principled way.

Functional Decision Diagrams (\fdd{}s), introduced by Kebschull et al.~\cite{KebschullFDD}, share some similarities with $\ldd$, starting with their names.
While, in both approaches, logic circuits are regarded (semantically) as Boolean functions,
only $\ldd$ regards reduction rules as functors operating on Boolean functions. 
Furthermore, $\ldd$ relies entirely on the Shannon operator (or combinator) to deconstruct Boolean functions whereas \fdd{} uses the positive Davio combinator. 
Nothing prevents using the latter in $\ldd$, and there is in fact a nice correspondence between the functors in both cases as detailed next. 

\subsection{Switching The Underlying Combinator}
\label{ax:fdd}

Recall that both Shannon and Davio combinators are \emph{universal} and \emph{elementary}:
universal means that any Boolean function can be expressed as a combination of constant functions;
elementary means that it operates on Boolean functions with the same arity by increasing the arity by $1$.

Becker and Drechsler~\cite{BeckerElemComb} identified a total of $12$ distinct universal and elementary combinators having the same form as Shannon's, except that the branching relies on an arbitrary function $p$ instead of the valuation of one (Boolean) variable.
By allowing output negation, they further reduced this number to only $3$, one of which is Shannon's;
the other two combinators are the positive and negative Davio combinators.
Recall that the (positive) Davio combinator is defined over Boolean functions of the same arity as
$f \star_{\texttt{D}^{+}} g$ as follows:
\[
    (f \star_{\texttt{D}^{+}} g): (x_0,x_{1},\dotsc,x_n) \mapsto f(x_{1},\dotsc,x_n) \oplus (x_0 \land g(x_{1},\dotsc,x_n)) \enspace .
\]

Devising a data structure where the combinator is a parameter would be very relevant to compare and better understand the benefits and drawbacks of switching the underlying combinator.
This would be a necessary first step towards a generic universal structure that allows even more complex combinators~\cite{ExtendedCofactoring,IntroBBDD,BertaccoDSD}.

In fact, Shannon-based reduction rules can be transposed into \emph{semantically equivalent} positive (or negative) Davio-based reduction rules.
To better appreciate this, let us detail an example.
To avoid any confusion, we use $\star_{\texttt{S}}$ below to denote the Shannon operator.
The following equation
\[
    f \star_{\texttt{S}} f = f \star_{\texttt{D}^{+}} 0
\]
holds for any Boolean function $f$ by definition of the combinators.
Hence, a useless variable for a positive Davio-based $\ldd$ would be syntactically captured by the
`same' introduction rule of the letter $\cc_{1\zero}$ we used for Shannon-based $\ldd$, leading to the following sameness relation denoted by `$\leftrightsquigarrow$'.
\[
    \inference{(\phi,n) \diamond_\texttt{S} (\blacksquare,n)}{\xrightarrow{\cc_{1\zero}} (\phi,n)} \leftrightsquigarrow \inference{(\phi,n) \diamond_{\texttt{D}^{+}} (\blacksquare,n)}{\xrightarrow{\uu} (\phi,n)}
\]
Table~\ref{tab:ShannonVsDavioReductionRule} summarizes this correspondence for the remaining elementary operators we have considered in this work.
This observation shows the flexibility of the $\ldd$ framework and settles the first steps towards extending it to support Davio operators.

\begin{table}
	\centering
	\caption{\label{tab:ShannonVsDavioReductionRule}Equivalence between Shannon-based reduction rules and Davio-based-reduction rules}
	\begin{tabular}{|ccc|}
		\hline
		$\texttt{S}$   & $\texttt{D}^{+}$ & $\texttt{D}^{-}$ \\ \hline
		$\uu$          & $\cc_{1\zero}  $ & $\cc_{1\zero}  $ \\
		$\xx$          & $\cc_{1\un}    $ & $\cc_{1\un}    $ \\ \hline
		$\cc_{0\zero}$ & $\cc_{0\zero}  $ & $\uu           $ \\
		$\cc_{0\un}$   & $\cc_{0\un}    $ & $\xx           $ \\ \hline
		$\cc_{1\zero}$ & $\uu           $ & $\cc_{0\zero}  $ \\
		$\cc_{1\un}$   & $\xx           $ & $\cc_{0\un}    $ \\ \hline
	\end{tabular}
\end{table}

\section*{Conclusion}
The functional point of view developed in this paper helps getting a better understanding of how different
existing variants of \bdd{} are related, by abstracting away several implementation details in order to solely focus on
how one constructs (or deconstructs) a Boolean function by adding (or removing) typed variable. 
This approach allowed us to propose a new data structure with clear functional semantics, and
to go beyond existing variants. 
We introduced a new model termed $\lddoNUCX$ that combines useless variables, \emph{all} canalizing variables and xor-variables while being invariant by negation. Its canonicity was formalized in the Coq proof assistant. 
More importantly, the approach we used could be very well instantiated using other elementary and arity-preserving operators that are application dependent achieving therefor a better compression rate. 


\bibliographystyle{siam}
\bibliography{refs}

\begin{thebibliography}{10}

\bibitem{IntroBBDD}
{\sc L.~Amarú, P.~Gaillardon, and G.~D. Micheli}, {\em Biconditional bdd: A
  novel canonical bdd for logic synthesis targeting xor-rich circuits}, in 2013
  Design, Automation Test in Europe Conference Exhibition (DATE), March 2013,
  pp.~1014--1017.

\bibitem{ESRBDD-TACAS2019}
{\sc J.~Babar, C.~Jiang, G.~Ciardo, and A.~Miner}, {\em Binary {Decision}
  {Diagrams} with {Edge}-{Specified} {Reductions}}, in Tools and {Algorithms}
  for the {Construction} and {Analysis} of {Systems}, T.~Vojnar and L.~Zhang,
  eds., {TACAS}'19, Springer International Publishing, 2019, pp.~303--318.

\bibitem{BeckerElemComb}
{\sc B.~Becker and R.~Drechsler}, {\em How many decomposition types do we need?
  [decision diagrams]}, in EDTC, 1995.

\bibitem{ExtendedCofactoring}
{\sc A.~Bernasconi, V.~Ciriani, G.~Trucco, and T.~Villa}, {\em On decomposing
  boolean functions via extended cofactoring}, in Proceedings of the Conference
  on Design, Automation and Test in Europe, DATE '09, 3001 Leuven, Belgium,
  Belgium, 2009, European Design and Automation Association, pp.~1464--1469.

\bibitem{BertaccoDSD}
{\sc V.~M. Bertacco}, {\em Achieving Scalable Hardware Verification with
  Symbolic Simulation}, PhD thesis, Stanford University, Stanford, CA, USA,
  2003.
\newblock AAI3104197.

\bibitem{Bryant1986}
{\sc R.~E. Bryant}, {\em Graph-based algorithms for boolean function
  manipulation}, IEEE Trans. Comput., 35 (1986), pp.~677--691.

\bibitem{CBDD+N}
\leavevmode\vrule height 2pt depth -1.6pt width 23pt, {\em Chain reduction for
  binary and zero-suppressed decision diagrams}, CoRR, abs/1710.06500 (2017).

\bibitem{CBDD}
{\sc R.~E. Bryant}, {\em Chain reduction for binary and zero-suppressed
  decision diagrams}, in {TACAS} {(1)}, vol.~10805 of Lecture Notes in Computer
  Science, Springer, 2018, pp.~81--98.

\bibitem{BurchLong1992}
{\sc J.~R. Burch and D.~E. Long}, {\em Efficient boolean function matching}, in
  1992 IEEE/ACM International Conference on Computer-Aided Design, Los
  Alamitos, CA, USA, 1992, IEEE Computer Society Press, pp.~408--411.

\bibitem{DarwicheClassification}
{\sc A.~Darwiche and P.~Marquis}, {\em A {Knowledge} {Compilation} {Map}}, 1,
  17 (2002), pp.~229--264.

\bibitem{canalising}
{\sc Q.~He and M.~Macauley}, {\em Stratification and enumeration of boolean
  functions by canalizing depth}, CoRR, abs/1504.07591 (2015).

\bibitem{KebschullFDD}
{\sc U.~Kebschull, E.~Schubert, and W.~Rosenstiel}, {\em Multilevel logic
  synthesis based on functional decision diagrams}, in [1992] {Proceedings}
  {The} {European} {Conference} on {Design} {Automation}, Mar. 1992,
  pp.~43--47.
\newblock ISSN: null.

\bibitem{knuth4A}
{\sc D.~Knuth}, {\em The Art of Computer Programming, Volume 4A: Combinatorial
  Algorithms, Part 1}, Pearson Education, 2014.

\bibitem{RolfVariantDual}
{\sc D.~M. Miller and R.~Drechsler}, {\em Dual edge operations in reduced
  ordered binary decision diagrams}, in Circuits and Systems, 1998. ISCAS '98.
  Proceedings of the 1998 IEEE International Symposium on, vol.~6, May 1998,
  pp.~159--162 vol.6.

\bibitem{ZDD}
{\sc S.~Minato}, {\em Zero-suppressed bdds for set manipulation in
  combinatorial problems}, in Proceedings of the 30th International Design
  Automation Conference, New York, NY, USA, 1993, ACM, pp.~272--277.

\bibitem{MinatoVariants}
{\sc S.~Minato, N.~Ishiura, and S.~Yajima}, {\em Shared binary decision diagram
  with attributed edges for efficient boolean function manipulation}, in 27th
  ACM/IEEE Design Automation Conference, Jun 1990, pp.~52--57.

\bibitem{IntroZDD}
{\sc A.~Mishchenko}, {\em An introduction to zero-suppressed binary decision
  diagrams}, tech. rep., in ‘Proceedings of the 12th Symposium on the
  Integration of Symbolic Computation and Mechanized Reasoning, 2001.

\bibitem{Somenzi1999}
{\sc F.~Somenzi}, {\em Efficient manipulation of decision diagrams},
  International Journal on Software Tools for Technology Transfer, 3 (2001).

\bibitem{TaggedBDD}
{\sc T.~van Dijk, R.~Wille, and R.~Meolic}, {\em Tagged bdds: Combining
  reduction rules from different decision diagram types}, in {FMCAD}, {IEEE},
  2017, pp.~108--115.

\end{thebibliography}

\clearpage
\appendix
\section{Normalizing $\lddoNUC$}
\label{ax:norm-negation}

The procedure $[\cdot]$ below reduces inductively a $\lddoNUC$ graph. 
$\ell$ denotes an elementary operator in $\Delta = \{\uu,\cc_{1\un},\cc_{1\zero},\cc_{0\un},\cc_{0\zero}\}$. 


\begin{lstlisting}
let $\phi_0$ $\nc{\diamond}$ $\phi_1$ =
    match $\phi_0 \diamond \phi_1$ with
    | $\phi \diamond \phi$ $\camlarrow$ $\xrightarrow{\uu}  \phi$
    | $\phi \diamond (\xrightarrow{\bullet}(\blacksquare,n))$ $\camlarrow$ $\xrightarrow{\cc_{1\un}}  \phi$
    | $\phi \diamond (\blacksquare,n)$ $\camlarrow$ $\xrightarrow{\cc_{1\zero}} \phi$
    | $(\blacksquare,n) \diamond (\xrightarrow{\bullet}\phi)$ $\camlarrow$ $\xrightarrow{\bullet.\cc_{0\un}} \phi$
    | $(\blacksquare,n) \diamond \phi$ $\camlarrow$ $\xrightarrow{\cc_{0\zero}} \phi$
    | $\_$ $\camlarrow$ $\phi_0\diamond\phi_1$
\end{lstlisting}

\begin{lstlisting}
let $\nc{\bullet}\phi$ =
  match $\phi$ with 
  | $\xrightarrow{\bullet} \phi'$ $\camlarrow$ $\phi'$
  | $\_$ $\camlarrow$ $\xrightarrow{\bullet}\phi'$ 
\end{lstlisting}

\begin{lstlisting}
let $\phi_0$ $\ncn{\diamond}$ $\phi_1$ =
    match $\phi_0$ with
    | $(\xrightarrow{\bullet}\phi_0')$ $\camlarrow$ $\nc{\bullet}(\phi_0'$ $\nc{\diamond}$ $(\nc{\bullet}\phi_1))$
    | $\_$ $\camlarrow$ $\phi_0$ $\nc{\diamond}$ $\phi_1$
\end{lstlisting}

\begin{lstlisting}
let rec  $[\phi]$ =
  match $\phi$ with
  | $(\blacksquare,0)$ $\camlarrow$ $(\blacksquare,0)$ 
  | $(\square,0)$ $\camlarrow$ $\xrightarrow{\bullet}(\blacksquare,0)$  
  | $\xrightarrow{\bullet} \phi'$ $\camlarrow$ $\nc{\bullet} [\phi']$ 
  | $\xrightarrow{\ell} \phi'$ $\camlarrow$ 
    let $\psi_0 \diamond \psi_1 = \texttt{elim-}\ell (\phi')$ in $[\psi_0]\ \ncn{\diamond}\ [\psi_1]$
  | $\phi_0 \diamond \phi_1$ $\camlarrow$ $[\phi_0]\ \ncn{\diamond}\ [\phi_1]$
\end{lstlisting}

\end{document}